\documentclass[10pt,oneside,letterpaper]{article}

\usepackage{graphicx,amsmath,amsfonts,amscd,amssymb,dcolumn,amsthm,textcomp,slashed,xspace,pgffor,tikz,latexsym,mathtools}
\usepackage[colorlinks,citecolor=blue,urlcolor=blue,linkcolor=blue]{hyperref}
\usepackage[T1]{fontenc}
\usepackage[utf8]{inputenc}
\usepackage{cite}
\usepackage{float}
\usepackage{xfrac}
\usepackage{footnote}
\usepackage{setspace}
\usepackage{soulutf8}
\usepackage{caption}
\usepackage{url}
\usepackage{pifont}
\usepackage{tikz-cd}
\usepackage[final]{pdfpages}
\usepackage[margin=0.92in]{geometry}
\usepackage{indentfirst}
\usepackage{xcolor}
\usepackage{mathrsfs}
\usepackage[hang]{footmisc}
\usepackage{tcolorbox}
\usepackage[shortcuts]{extdash}
\usepackage{scalefnt}
\usepackage{cuted}
\usepackage{booktabs}
\usepackage{lipsum}
\tikzcdset{arrow style=tikz,
diagrams={>={Straight Barb[scale=0.8]}}}

\tcbset{boxrule=0pt,colframe=white,colback=green!60!,arc=0pt,left=0.1pt,right=0.1pt,top=0.1pt,bottom=0.1pt}

\numberwithin{equation}{section}

\theoremstyle{plain}
\newtheorem{theorem}{Theorem}[section]
\newtheorem{corollary}{Corollary}[section]
\newtheorem{lemma}{Lemma}[section]
\newtheorem{definition}{Definition}[section]
\newtheorem{proposition}{Proposition}[section]
\newtheorem{remark}{Remark}[section]

\begin{document}

\title{\textbf{Homological aspects of topological gauge-gravity equivalence}}

\author{\textbf{Thiago S.~Assimos}$^{1}$\thanks{thiago.assimos@gmail.com}\,\, and \textbf{Rodrigo F.~Sobreiro}$^1$\thanks{rodrigo\_sobreiro@id.uff.br}\\\\
\textit{{\small $^1$UFF - Universidade Federal Fluminense, Instituto de Física,}}\\
\textit{{\small Campus da Praia Vermelha, Av. Litorânea s/n, 24210-346,}}\\
\textit{{\small Niterói, RJ, Brasil.}}}
\date{}
\maketitle

\begin{abstract}
In the works of A. Achúcarro and P. K. Townsend and also by E. Witten, a duality between three-dimensional Chern-Simons gauge theories and gravity was established. In all cases, the results made use of the field equations. In a previous work, we were capable to generalize Witten's work to the off-shell cases, as well as to four dimensional Yang-Mills theory with de Sitter gauge symmetry. The price we paid is that curvature and torsion must obey some constraints under the action of the interior derivative. These constraints implied on the partial breaking of diffeomorphism invariance. In the present work, we, first, formalize our early results in terms of fiber  bundle theory by establishing the formal aspects of the map between a principal bundle (gauge theory) and a coframe bundle (gravity) with partial breaking of diffeomorphism invariance. Then, we study the effect of the constraints on the homology defined by the interior derivative. The main result is the emergence of a nontrivial homology in Riemann-Cartan manifolds.

\end{abstract}

\maketitle

\section{Introduction}

In the work \cite{Achucarro:1987vz}, A. Achúcarro and P. K. Townsend demonstrated the relationship between an $SO(p,q)$-anti-de Sitter supergravity theory and the three-dimensional $OSp(p|2;\mathbb{R})\otimes OSp(q|2;\mathbb{R})$ theory. Furthermore, they also showed that, when a connection is defined for the supergroup, then the integral of the Chern-Simons (CS) three-form is equivalent to the supergravity theory. In this work the Poincaré limit is obtained from an In\"on\"u-Wigner contraction \cite{Inonu:1953sp}: from the rescaling of the gauge field $A\rightarrow mA$ ($m$ is a mass parameter) followed by the limit
$m\rightarrow0$. The result being a Poincaré supergravity. Latter, in \cite{Witten:1988hc}, E. Witten also discusses the equivalence between CS theory for the Poincaré group (and also de Sitter symmetry) and three-dimensional Einstein-Hilbert (EH) gravity. However, the local isometries and diffeomorphisms of the gravity theory are obtained from the gauge symmetry of the CS original theory. Hence, the Inönü-Wigner contraction is actually not required. Here, it is noteworthy that in both cited papers, namely \cite{Achucarro:1987vz,Witten:1988hc}, the field equations are extensively employed in order to they achieve a consistent gravity theory. Other examples of gravity theories obtained from gauge theories can be found in, for instance, \cite{Obukhov:1998gx,Sobreiro:2007pn,Sobreiro:2011hb,Assimos:2013eua}.

A few of decades after the seminal papers \cite{Achucarro:1987vz,Witten:1988hc}, in \cite{Assimos:2019yln}, we have extended their gauge-gravity equivalence (GGE) analysis in two ways: First, we were able to abolish the need of the field equations, performing an off-shell generalization. The price we paid for that is to deal with a pair of constraints over curvature and torsion 2-forms. Because these constraints define spacetime foliations, they imply on the breaking of the diffeomorphism invariance of the theory. Second, we generalize the GGE to four dimensions where we consider the orthogonal group as the gauge group of a Yang-Mills theory and employed the same ideas of \cite{Witten:1988hc} by showing that the gauge group can generate local Lorentz isometries as well as diffeomorphisms.

In the present work we review the results of \cite{Assimos:2019yln} within a formal mathematical apparatus. First, the physical and mathematical structures of the GGE is extensively exploited under the eyes of fiber bundle theory \cite{Kobayashi:1963fg,Daniel:1979ez,Nash:1983cq,Nakahara:2003nw}. We achieve this goal by considering three examples: Three-dimensional CS theory for the Poincaré group; Three-dimensional CS theory for the $SO(4)$ group; Four-dimensional Pontryagin theory for the $SO(5)$ group. All resulting gravity theories will suffer from a partial breaking of diffeomorphism invariance due to the constraints over curvature and torsion. 

Second, we show how these constraints affect the homology groups of Riemann-Cartan (RC) manifolds associated with the interior derivative nilpotent operator. This operator indeed appear, for instance, in the constraints of gravity obtained from Poincaré CS gauge theory, imposing that curvature and torsion are closed 2-forms. The homology group of the interior derivative over a RC manifold are inferred by explicitly computing the nontrivial $q$-cycles of the interior derivative subjected to above referred constraints. For simplicity, such homology is here called \emph{interior homology}. Essentially, the interior homology of free (with no constraints) RC manifolds is trivial. The constraints will then bring non-trivialities to the interior homology of RC manifolds.

The novelty of the present paper are contained in Theorems \ref{theo.cofr1}, \ref{theo.cofr2} and \ref{theo.cofr3} and the whole Section \ref{sec.const.inthomol}. The rest of the paper are necessary for helpful definitions, self-consistency and completeness.

This work is organized as follows: In Sec.~\ref{sec.mp} the most relevant concepts of fiber bundles, connection theory, group theory and differential geometry will be formally used in order to construct the gauge and gravity theories. In Sec.~\ref{sec.gge} the GGE is discussed in three and four dimensions also through a robust mathematical setup. In this section, we actually formalize the GGE under bundle maps theory. In Sec.~\ref{sec.inthomol} we discuss the interior homology of free RC manifolds. Then, in Sec.~\ref{sec.const.inthomol} nontrivial interior homology groups of RC manifolds subjected to the constraints over curvature and torsion are explicitly computed. Finally, in Sec.~\ref{sec.concl} our conclusions and perspectives are outlined.

\section{Mathematical preliminaries}\label{sec.mp}

We start with some mathematical and physical definitions which will be employed along the paper. In this section, we assume a generic $n$-dimensional spacetime. Latter on, we will particularize for three and four dimensions. We will discuss the geometrical setting of gauge theories and gravitation. Hence, in the next section we will study the mapping of one theory into another within specific examples.

\subsection{Gauge theories}\label{sub.gt}

We start by constructing a gauge theory in the fiber bundle scenario.
\begin{definition}\label{rcm1}
An $n$-dimensional Riemann-Cartan manifold $M$ exists. The manifold $M$ is assumed to be a paracompact Hausdorff manifold.
\end{definition}

The most basic mathematical object one wishes to define in the construction of a field theory is, perhaps, spacetime. The manifold $M$ represents spacetime. The Hausdorff and paracompactness properties are required since in such manifolds: points are separable; a metric can be defined; and a partition of unity exists. See for instance \cite{Kobayashi:1963fg,Wald:1984rg}. In $M$ one can define basic physical fields and particles. A classical gauge theory, on the other hand, needs some further sophisticated mathematical structures such as a principal bundle and the corresponding connection 1-form \cite{Daniel:1979ez,Nakahara:2003nw,Bertlmann:1996xk}. Hence, a Lie group must be introduced at this point.

\begin{definition}\label{Lie1}
Let $\mathbb{G}$ be a Lie group endowed with a stability group $S$ such that $C=\mathbb{G}/S$ is a symmetric coset\footnote{$C$ is a symmetric coset if
\begin{eqnarray}
    \left[L,L\right]&\sim&L\;,\nonumber\\
    \left[L,Q\right]&\sim&Q\;,\nonumber\\
    \left[Q,Q\right]&\sim&L\;\mathrm{or}\;0\;.
\end{eqnarray}
In the case of $0$ in the last relation, the coset is an Abelian symmetric subgroup.}. The generators of the algebra of $\mathbb{G}$ are denoted by $L_A$ and $Q_{\bar{A}}$ such that Caption Latin indices run through $\{1,2,\ldots,\mathrm{dim}\,S\}$ while Caption Latin indices with a bar run through $\{1,2,\ldots,\mathrm{dim}\,C\}$.
\end{definition}

In possession of definitions \eqref{rcm1} and \eqref{Lie1}, we can define the principal bundle describing the gauge theory as well as the connection representing the gauge field.

\begin{definition}\label{fiberbundle1}
Let $P=\{\mathbb{G},M,\pi,\phi\}$ be a principal bundle, where $P$ is the total space, $\mathbb{G}$ is both the fiber and the structure Lie group, $M$ is the base space, $\pi:P\rightarrow M$ is a continuous surjective projection map and $\pi^{-1}(x)$ is an inverse image at a point $x\in M$. The diffeomorphisms $\phi_{\alpha}:U_{\alpha}\times \mathbb{G} \rightarrow\pi^{-1}(U_{\alpha})$ are local trivializations, where $U_{\alpha}$ are open sets covering $M$. Furthermore $\mathrm{t}_{\alpha\beta}(x)=\phi^{-1}_{\alpha,x}\circ\phi_{\beta,x}:\mathbb{G}\rightarrow \mathbb{G}$ are the transition functions.
\end{definition}

\begin{definition}\label{def.connection}
The principal bundle $P$ is endowed with an algebra-valued connection $\mathrm{1}$-form, $H=A^AL_A+\theta^{\bar{A}}Q_{\bar{A}}$.
\end{definition}

Locally the connection 1-form $H$ is obviously recognized as an algebra-valued gauge field. The fields $A^A$ and $\theta^{\bar{A}}$ are the fundamental fields of the gauge theory we are defining.

\begin{definition}\label{def.fieldstrenght}
Let $K$ be the local curvature $\mathrm{2}$-form on $P$. Hence\footnote{Clearly, the explicit forms of $F$ and $\Pi$ depend on the algebra of $\mathbb{G}$. Later on we will specify them for the relevant groups explored in this paper.}, $K=dH+HH=F^AL_A+\Pi^{\bar{A}}Q_{\bar{A}}$, where
\begin{eqnarray}
F^A&=&dA^A+\Sigma^A_{\phantom{A}BC}A^BA^C+\Sigma^A_{\phantom{A}\bar{A}\bar{B}}\theta^{\bar{A}}\theta^{\bar{B}}\;,\nonumber\\
\Pi^{\bar{A}}&=&D\theta^{\bar{A}}\;\;=\;\;d\theta^{\bar{A}}+\Sigma^{\bar{A}}_{\phantom{\bar{A}}A\bar{B}}A^A\theta^{\bar{B}}\;,\label{tf}
\end{eqnarray}
with $d$ being the exterior derivative in $M$, $D=d+A$ is the covariant derivative with respect to the stability group connection, and the $\Sigma$-tensors are the decomposed structure constants of the algebra of $\mathbb{G}$.
\end{definition}

The field $K$ is clearly identified as the field strength in a gauge theory. The components $F^A$ and $\Pi^{\bar{A}}$ will be interpreted in the following sections because they depend on the group structure.

\begin{definition}\label{def.gaugetransf0}
An infinitesimal gauge transformation $\delta_g$ is an automorphism $P\rightarrow P$, \textit{i.e.}, takes points of the fiber in other points of the same fiber, such that, 
\begin{eqnarray}\label{gt0}
\delta_g A^A&=&D\alpha^A+\Sigma^A_{\phantom{A}\bar{A}B}\theta^{\bar{A}}\zeta^B\;, \nonumber \\
\delta_g\theta^{\bar{A}}&=&D\zeta^{\bar{A}}+\Sigma^{\bar{A}}_{\phantom{{\bar{A}}}A{\bar{B}}}\alpha^A\theta^{\bar{B}}\;, 
\end{eqnarray}
with $\varrho=\alpha^AL_A+\zeta^{\bar{A}}Q_{\bar{A}}$ being the infinitesimal parameter of the gauge group $\mathbb{G}$. Expressions \eqref{gt0} define the gauge orbits of $P$.
\end{definition}

\begin{figure}[ht!]
\centering
\includegraphics[width=0.8\linewidth]{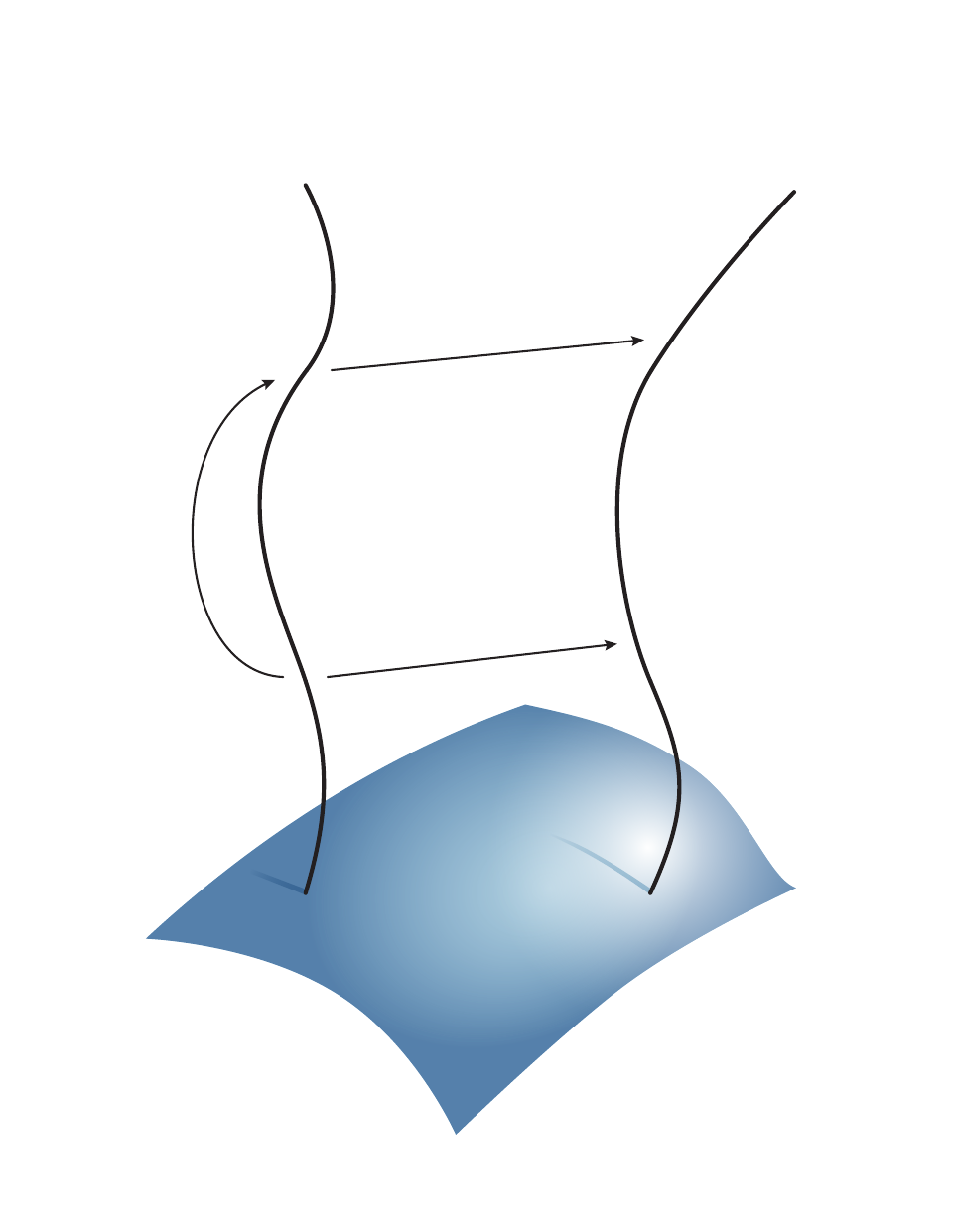}
\put(-318,283){$\delta_g$}
\put(-200,233){$H$}
\put(-200,356){$H'$}
\put(-269,130){$x$}
\put(-133,129){$x+\delta x$}
\put(-215,33){$M$}
\put(-269,428){$\mathbb{G}(x)$}
\put(-89,427){$\mathbb{G}(x+\delta x)$}
\caption{Schematics of the structure of the principal bundle $P$ describing a gauge theory. The blue surface represents spacetime and the lines represent the fibers $\mathbb{G}(x)$.}
\label{fig1}
\end{figure}
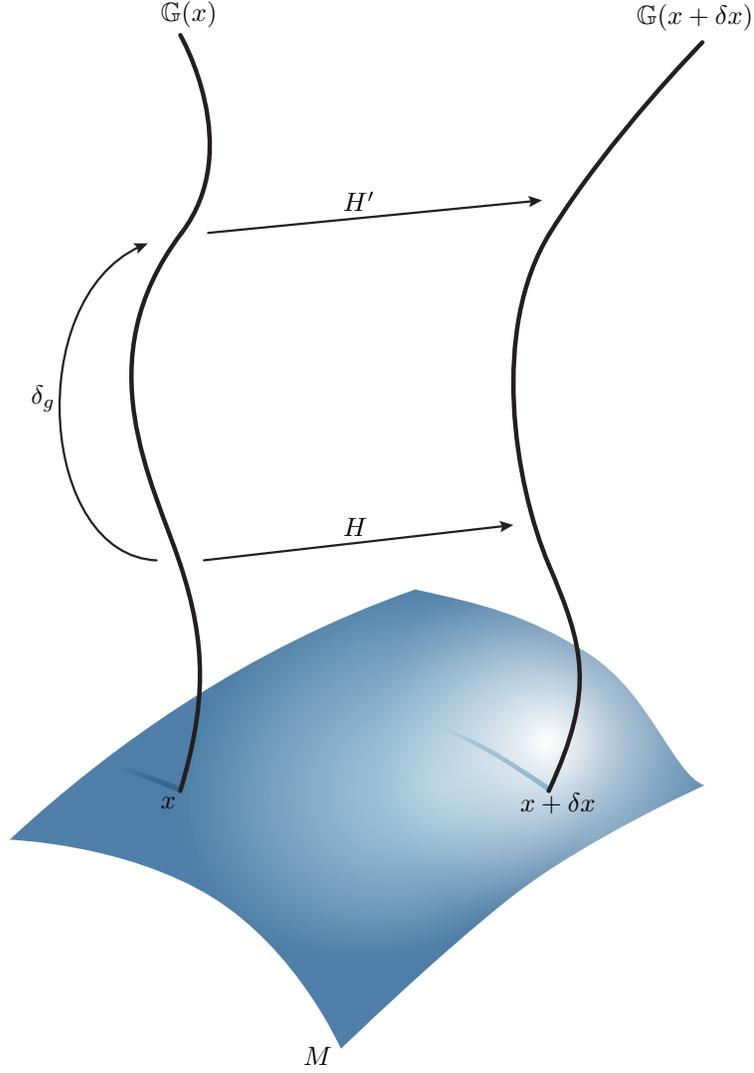

\begin{proposition}\label{staticM}
The manifold $M$ is nondynamical.
\end{proposition}
In physical terms, $M$ is an inanimate stage for the physical entities to live, \emph{i.e.}, the background spacetime carrying no dynamics whatsoever. Ultimately, this means that the gauge theory we are constructing in $P$ carries no gravitational degrees of freedom. Hence, a gauge theory is a dynamical set of equations for the gauge field $H$ ($A$ and $\theta$) immersed in a background nondynamical spacetime $M$. See Figure \ref{fig1} for a schematic representation of $P$ and its structure.

\subsection{Gravity as a gauge theory}\label{sub.ggt}

It is a known fact that gravity can be dressed as a gauge theory, see for instance \cite{Utiyama:1956sy,Kibble:1961ba,Sciama:1964wt,Hehl:1976vr,Nakahara:2003nw,McInnes:1984kz,deSabbata:1986sv,Nash:1983cq,Zanelli:2005sa,Bertlmann:1996xk}. All what is needed is the identification of the structure gauge group with the local isometries of the base space. This identification is realized by the introduction of the vielbein\footnote{The word \emph{vielbein} is generally used to specify the soldering form $e$ in $n$ dimensions. In the case of 3 and 4 dimensions, the soldering form is called dreibein and vierbein, respectively.} field on the base space of a specific principal bundle. Such identification gives rise to the frame and coframe bundles.

\begin{proposition}\label{son1}
An $n$-dimensional Riemann-Cartan manifold $M^\prime$ exists. Its local isometries are characterized by the $SO(n)$ compact Lie group. The adjoint generators of the group $L_{\mathfrak{ab}}$ are antisymmetric in their indices and the fundamental generators are $P_{\mathfrak{a}}$, such that the small Latin indices of the type $\mathfrak{a,b,c}\ldots$ run through $\{1,2,\ldots,n\}$.
\end{proposition}

\begin{definition}\label{vielbein1}
The vielbein $\mathrm{1}$-form field, $e=e^\mathfrak{a}P_\mathfrak{a}$,
is a local isomorphism $M^\prime\longmapsto T^\ast(M^\prime)$, with $T^\ast(M^\prime)$ being the cotangent space at a point $x'\in M^\prime$.
\end{definition}

Physically, the vielbein ensures that, locally, one can always find a local inertial frame  $dx^\mathfrak{a}$ out from a general frame\footnote{Lower case Greek indices refer to spacetime indices running through $\{0,1,\ldots,n\}$.} $dx^\mu$, \emph{i.e.} $dx^{\mathfrak{a}}=e^{\mathfrak{a}}_\mu dx^\mu$. Hence, local inertial frames are identified with the tangent space at a point $x'\in M^\prime$. Therefore, one can readily see that the equivalence principle \cite{deSabbata:1986sv,Wald:1984rg,Misner:1973prb} is ensured by the existence of the vielbein field. Moreover, since $SO(n)$ is the local isometry group, there is an infinity number of equivalent vielbeins at a point $x'\in M^\prime$, meaning that there is an infinity number of inertial frames at our disposal. The most basic gravity theory is realized by making $e$ a dynamical field.

\begin{definition}\label{cot.bundle1}
Let $P^\prime=\{SO(n),M^\prime,\pi,\phi\}$ be a principal bundle called cotangent bundle. The total space is $P^\prime$, the structure group is $SO(n)$, and the base space is $M^\prime$. The structure group $SO(n)$ is also the local isometry group of $M^\prime$. The projection $\pi:P^\prime\rightarrow M^\prime$ is a continuous surjective and its inverse, $\pi^{-1}(x^\prime)$, is an inverse image at a point $x^\prime\in M^\prime$. The diffeomorphisms $\phi_{\alpha}:U_{\alpha}\times SO(n) \rightarrow\pi^{-1}(U_{\alpha})$ ensures local trivialization with $U_{\alpha}$ being open sets covering $M^\prime$. Moreover, $\mathrm{t}_{\alpha\beta}(x^\prime)=\phi^{-1}_{\alpha,x^\prime}\circ\phi_{\beta,x^\prime}:GL(n,\mathbb{R})\rightarrow GL(n,\mathbb{R})$ are the transition functions.
\end{definition}

The cotangent bundle can be understood as follows: At $T^\ast(M^\prime)$ one defines all frames $e$ as the fiber at $x'\in M^\prime$. Since all frames can be obtained from the action of an $SO(n)$ element, the fiber is also the group $SO(n)$. Another way to see it is as a principal bundle (gauge theory) where the gauge group is identified with the local isometries of spacetime by defining the vielbein field. Moreover, the structure group is extended to $GL(n,\mathbb{R})$. Such extension does not affect geometric and topological properties of $P'$ because the general linear group is contractible to the orthogonal group \cite{Nash:1983cq,McInnes:1984kz,Sobreiro:2010ji}.

\begin{definition}\label{def.spin.connection}
The principal bundle $P^\prime$ is endowed with a local algebra-valued connection $\mathrm{1}$-form, $\omega=\omega^\mathfrak{ab}L_\mathfrak{ab}$, typically called spin-connection.
\end{definition}

In RC manifolds, $e$ and $\omega$ are independent gravity variables. While the vielbein characterizes the metric properties of $M^\prime$, the spin-connection features the parallelism properties of spacetime.

\begin{definition}\label{def.RT}
Let $\Omega=\Omega^{\mathfrak{ab}}L_{\mathfrak{ab}}$ be the local curvature $\mathrm{2}$-form and $T=T^{\mathfrak{a}}P_{\mathfrak{a}}$ the local torsion $\mathrm{2}$-form on $P^\prime$. Hence, $\Omega^{\mathfrak{ab}}=d\omega^{\mathfrak{ab}}+\omega^{\mathfrak{ac}}\omega_{\mathfrak{c}}^{\phantom{\mathfrak{a}}\mathfrak{b}}$ and $T^{\mathfrak{a}}=\nabla e^{\mathfrak{a}}$ with the full covariant derivative defined by $\nabla e^{\mathfrak{a}}=de^{\mathfrak{a}}+\omega^{\mathfrak{a}}_{\phantom{\mathfrak{a}}\mathfrak{b}}e^{\mathfrak{b}}$ and the derivative $d$ is the exterior derivative in $M^\prime$.
\end{definition}

\begin{definition}\label{def.gaugetransf2} The automorphism $P'\rightarrow P'$ is an infinitesimal gauge transformation $\delta_I$ such that,
\begin{eqnarray}\label{gt1a}
\delta_I \omega^{\mathfrak{ab}}&=&\nabla\alpha^{\mathfrak{ab}}\;, \nonumber \\
\delta_Ie^{\mathfrak{a}}&=&\alpha^{\mathfrak{a}}_{\phantom{\mathfrak{a}}\mathfrak{b}}e^{\mathfrak{b}}\;,\label{gt3}
\end{eqnarray}
with $\alpha=\alpha^{\mathfrak{ab}}L_{\mathfrak{ab}}$ being the infinitesimal parameter of the gauge group $SO(n)$. 
\end{definition}
The gauge transformations in Definition \ref{def.gaugetransf2} characterize the local spacetime isometries. Hence, the present formalism puts gravity as a gauge theory by gauging the local isometries. A schematic figure illustrating the coframe bundle and its structure is displayed in Figure \ref{fig2}.

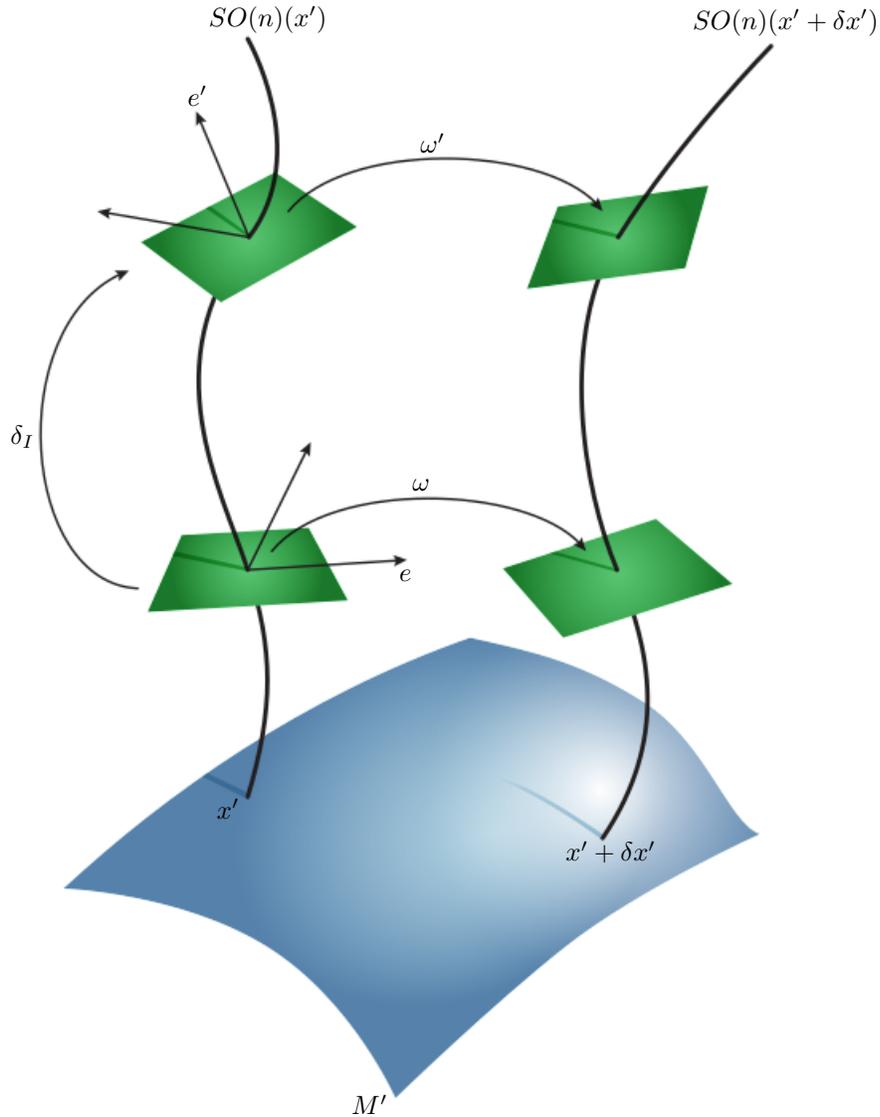
\begin{figure}[ht!]
\centering
\includegraphics[width=0.8\linewidth]{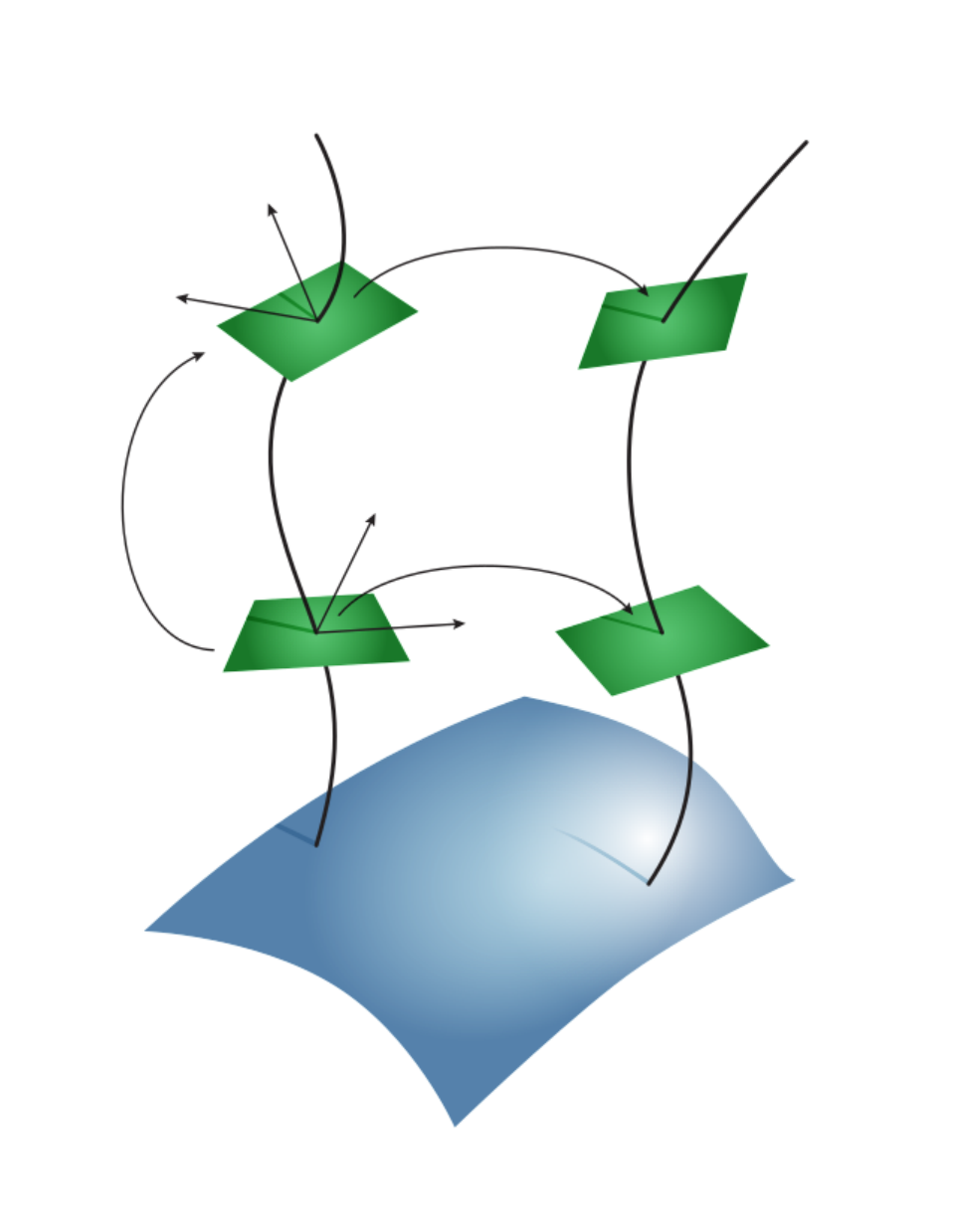}
\put(-347,290){$\delta_I$}
\put(-192,400){$\omega'$}
\put(-195,272){$\omega$}
\put(-200,238){$e$}
\put(-269,148){$x'$}
\put(-137,132){$x'+\delta x'$}
\put(-218,36){$M'$}
\put(-272,448){$SO(n)(x')$}
\put(-280,418){$e'$}
\put(-89,447){$SO(n)(x'+\delta x')$}
\caption{Schematics of the structure of the coframe bundle $P^\prime$ describing a gravity theory. The blue surface represents spacetime and the lines represent the fibers $SO(n)$. The green surfaces represent the collection of all tangent spaces that can be reached from a vielbein map.}
\label{fig2}
\end{figure}

\subsection{Further definitions}\label{sub.fd}

Two important operators will be of great relevance in the present work, the interior derivative and the Lie derivative, both defined in the sequence. 

\begin{definition}\label{int.deriv}
Let $\mathfrak{X}(M)$ be a vector space on Riemann-Cartan $n$-dimensional manifold $M$ and $\mathcal{V}^q(M)$ the vector space of smooth $q$-forms $\in$ $M$. The interior derivative with respect to $X\in \mathfrak{X}(M)$ is the unique antiderivation of degree $-1$ denoted by $\mathrm{i}_X$ such that if $\Psi\in\mathcal{V}^0(M)$ and $\Psi\in C^{\infty}(M)$ then $\mathrm{i}_X\Psi=0$ and $\mathrm{i}_X\beta=\beta(X)$ with $\beta\in\mathcal{V}^1(M)$. Therefore for a $q$-form $\xi$, $\mathrm{i}_X\xi(X_1,\dots,X_{q-1})\equiv\xi(X_,X_1,\dots,X_{q-1})$. In other words $\mathrm{i}_X\xi=\xi(X)$ is obtained from $\xi$ by inserting $X$ into the first slot and so on for differential forms of higher degree. 
\end{definition}

\begin{remark}\label{nilpotency1}
Some immediate properties of the interior derivative $\mathrm{i}_X$ are:
\begin{enumerate}
\item Linearity: $\mathrm{i}_{X+Y}=\mathrm{i}_X+\mathrm{i}_Y$, for $X$ and $Y$ $\in\mathfrak{X}(M)$.

\item Leibniz rule: $\mathrm{i}_X(\lambda\sigma)=\mathrm{i}_X\lambda\sigma+(-1)^{q}\lambda\mathrm{i}_X\sigma$, with $\lambda\in\mathcal{V}^q(M)$ and $\sigma\in\mathcal{V}^p(M)$.

\item Nilpotency: $\mathrm{i}_X\circ\mathrm{i}_X=\mathrm{i}_{X}^2=0$.
\end{enumerate}
\end{remark}

\begin{definition}\label{def.liediff}
Given the one-parameter diffeomorphism group $\varphi_{s}:M\rightarrow M$, then $\varphi_{s}$ is a diffeomorphic commutative map. 
\end{definition}

\begin{definition}\label{def.derivlie1}
Let $\bar{x}$ be a nearby point\footnote{Is the important remember that $\bar{x}$ must belong to the diffeomorphism generated by $X$, \textit{i.e.}, $\bar{x}=\varphi_s(x).$} of $x$, then the Lie derivative of a $q$-form $\xi$ along the diffeomorphism $\varphi_s$ associated to $X$ is given by the difference between the pullback of $\xi(\bar{x})$ to the point $x$ by $\varphi^{\ast}_{-s}$ and $\xi(x)$,   
\begin{equation}\label{diff1}
\mathcal{L}_X\xi=\displaystyle\lim_{s\to 0}\dfrac{\varphi^{\ast}_{-s}\xi(\bar{x})-\xi(x)}{s}\,,
\end{equation}
and the Lie derivative satisfies $\mathcal{L}_X\cdot=\mathrm{i}_X(d\cdot)+d(\mathrm{i}_X\cdot)$.
\end{definition}

\section{Gauge-gravity equivalence}\label{sec.gge}

In this section we discuss the results obtained in \cite{Assimos:2019yln} in a more formal way. We will focus mainly in the three-dimensional Poincaré case for the CS gauge theory. The case for S$O(4)$ symmetry in three-dimensional CS theory and the four-dimensional orthogonal case for the Pontryagin gauge theory will be only briefly outlined because the recipe is essentially the same. The main idea is to establish a map between a topological gauge theory and a gravity theory with constrained diffeomorphisms by means of a map $P\longmapsto P^\prime$ and also at the level of the dynamical actions that we will define in this section. 

The results in this section generalize some of the results in the seminal works of Achúcarro and Townsend \cite{Achucarro:1987vz} and of Witten \cite{Witten:1988hc} to the off-shell case as well as to the four-dimensional case. Moreover, they improve the results already obtained in \cite{Assimos:2019yln}.

\subsection{Three-dimensional Poincaré Chern-Simons theory}\label{sub.cs3d}

We consider a gauge theory in the principal bundle in Definition~\ref{fiberbundle1} for the gauge group in Definition ~\ref{Lie1} taken as the Poincaré group in the representation $SU(2)\times\mathbb{R}^3$. Hence, $SU(2)$ is the stability group and $\mathbb{R}^3$ is the coset space (which in this case is an Abelian subgroup). Moreover, spacetime is taken as a three-dimensional RC manifold, $\mathrm{dim}\,M=3$. Thence\footnote{From now on, we omit the projection and the trivializations for the sake of simplicity.}, $P=\{SU(2)\times\mathbb{R}^3,M\}$. Before we map $P$ in a gravity cotangent bundle, we need to establish more specifics of this example.

\begin{definition}\label{algebra}
Let $isu(3)$ be the Lie algebra of the Poincaré group in the representation $ISO(3)\equiv SU(2)\times\mathbb{R}^3$. Let $L_i$ be the generators of the $SU(2)$ sector and $Q_i$ the translational generators. Then, $isu(3)$ is realized through
\begin{eqnarray}\label{alg1}
\left[L_{i},L_{j}\right]&=&{\epsilon_{ijk}}L^{k}\;,\nonumber\\
\left[L_{i},Q_{j}\right]&=&{\epsilon_{ijk}}Q^{k}\;, \nonumber \\ 
\left[Q_{i},Q_{j}\right]&=& 0\;,
\end{eqnarray}
where lower case Latin indices from $i$ to $z$ vary through $\{0,1,2\}$. The Killing metric is given by 
\begin{eqnarray}\label{iqf1}
\mathrm{Tr}(L^{i}L^{j})&=&0\;,\nonumber\\
\mathrm{Tr}(Q^{i}Q^{j})&=&0\;,\nonumber\\
\mathrm{Tr}(L^{i}Q^{j})&=&\eta^{ij}\;, 
\end{eqnarray}
with $\mathrm{Tr}$ being a quadratic invariant form in the group and the metric $\eta^{ij}$ and the Levi-Civita symbol $\epsilon_{ijk}$ are invariant tensors in group space.
\end{definition}

\begin{definition}\label{def.gaugetransf1}
The infinitesimal gauge transformations $\delta_g$ are
\begin{eqnarray}\label{gt1}
\delta_g A^i&=&D\alpha^i\;, \nonumber \\
\delta_g\theta^i&=&D\zeta^i+\epsilon^i_{\phantom{i}jk}\alpha^j\theta^k\;, 
\end{eqnarray}
with $\varrho=\alpha^{i}L_{i}+\zeta^{i}Q_{i}$ being the infinitesimal gauge parameter of the gauge group.
\end{definition}

\begin{definition}\label{def.action1}
The Chern-Simons action, providing dynamics to the gauge field $H$, invariant under \eqref{gt1}, is defined by
\begin{equation}\label{cs1}
S_{CS} = \dfrac{m}{4\pi\kappa^2}\mathrm{Tr}\int\left(H dH + \frac{2}{3}HHH\right)\;,
\end{equation}
with $m$ being the Chern-Simons topological mass and $\kappa$ the coupling constant. 
\end{definition}

\begin{remark}\label{mdim}
The mass dimensions of the gauge field and parameters are $[H]=1/2$, $[\kappa]=1/2$, and $[m]=1$. 
\end{remark}

Being metric independent, the CS action is indeed of topological nature \cite{Piguet:1995er}. Moreover, one can easily check that the CS action is invariant under transformations \eqref{gt1}.

\begin{proposition}\label{prop.csdecomp}
Employing decomposition \eqref{def.connection} and using \eqref{alg1} and \eqref{iqf1}, the Chern-Simons action \eqref{cs1} is rewritten as
\begin{equation}\label{act1}
S_{CS}=\dfrac{m}{2\pi\kappa^2}\int\theta^{i}F_{i}\;, 
\end{equation}
with $F^i=dA^i+\frac{1}{2}\epsilon^i_{\phantom{i}jk}A^jA^k$ (See Definition~\ref{def.fieldstrenght}).
\end{proposition}

\begin{proposition}\label{lemma.feq1}
The variation of the action \eqref{act1} with respect to the gauge fields $A^i$ and $\theta^i$ leads to the trivial field equations $F^i=\Pi^i=0$, with $\Pi^i=D\theta^i$.
\end{proposition}

The gauge theory is now fully constructed as described by the CS action \eqref{act1}. We are now ready to map it into a gravity theory.

\begin{proposition}\label{prop.mapgaugegrav}
Let $SU(2)\longmapsto SO(3)$ be a surjective homomorphism used to map the gauge theory in the gravity theory, then the corresponding identifications of generators are $L_i=\frac{1}{2}\epsilon_{i}^{\phantom{i}jk}L_{jk}$, with $L_{jk}$ being the generators of the $SO(3)$ group.
\end{proposition}

\begin{proposition}\label{prop.gaugemap}
The $SU(2)\longmapsto SO(3)$ homomorphism is used to map the fields of the gauge theory in gravity fields by means of
\begin{eqnarray}\label{id1}
A^i&\longmapsto&-\frac{1}{2}\epsilon^{ijk}\omega_{jk}\;,\nonumber \\
\theta^i&\longmapsto&\mu e^{i}\;,
\end{eqnarray}
where $\mu$ is any parameter with half mass dimension $[\mu]=1/2$.
Clearly, the parameter $\mu$ is needed because the dreibein carries no mass dimension, $[e]=0$. In practice $\mu$ can be any function of $\kappa$ and $m$ with the correct dimension, $\mu=h(\kappa,m)\;\big|\;[h]=1/2$.
\end{proposition}

\begin{proposition}\label{lemma.mapaction}
The three dimensional Einstein-Hilbert action is obtained by employing the map \eqref{id1} into the action \eqref{act1}. Then,
\begin{equation}\label{eh1}
S_{CS}\longmapsto S_{3grav}=-\frac{1}{16\pi G}\int\epsilon_{ijk}e^{i}\Omega^{jk}\;,
\end{equation}
where $G=\dfrac{\kappa^2}{8m\mu}$ is the three-dimensional Newton's constant and $\Omega^{ij}=d\omega^{ij}+{\omega^i}_k\omega^{kj}$ is the gravitational curvature $\mathrm{2}$-form in three dimensions. 
\end{proposition} 

\begin{proposition}\label{lemmafeq2}
The variation of the action \eqref{eh1} with respect to $e^i$ and $\omega^{ij}$ results in the field equations $\Omega^{ij}=T^i=0$.
\end{proposition}
 
\begin{proposition}\label{lemma.gaugetransf2}
The identifications \eqref{id1} when applied on the gauge transformations \eqref{gt1} induces a local gauge transformation on the fields in the form
\begin{eqnarray}\label{gt2}
\delta_{g}\omega^{ij}&=&\nabla\alpha^{ij}\;, \nonumber \\
\delta_g e^i&=&\nabla\vartheta^{i}+{\alpha^i}_je^j\;,
\end{eqnarray}
with $\alpha^{ij}=-\epsilon^{ijk}\alpha_{k}$ and $\vartheta^i=\frac{1}{\mu}\zeta^i$.
\end{proposition}

At this point we have obtained a gravity theory described by the action \eqref{eh1} from the CS action \eqref{cs1}. Such result was only possible due to identifications \eqref{id1} which defines the "absorption" of the gauge fields into spacetime. In other words, a sector of the gauge group is identified with the local isometries and the field $A$ identified with the spin-connection. The identification of the fields $A$ and $\theta$ with the spin-connection and the dreibein, automatically induces dynamics to the spacetime. It remains however to split the resulting gauge symmetry \eqref{gt2} in order to actually obtain the symmetries of a coframe bundle.

\begin{proposition}\label{lemma.liederiv}
The Lie derivatives of the gravitational fields $\omega^{ij}$ and $e^i$ are
\begin{eqnarray}\label{diff3}
\mathcal{L}_X\omega^{ij}&=& \mathrm{i}_X\Omega^{ij}+\nabla(\mathrm{i}_X\omega^{ij})\;,\nonumber\\
\mathcal{L}_Xe^{i} &=&\mathrm{i}_XT^i+\nabla(\mathrm{i}_Xe^i)-(\mathrm{i}_X{\omega^i}_k)e^k\;.
\end{eqnarray}
\end{proposition}

\begin{lemma}\label{theorem.isometr}
The local isometries ($SO(3)$ transformations) in $M^\prime$ can be obtained from $\mathcal{L}_X$ and $\delta_{_{\mathbb{R}^3}}$, up to suitable constraints over curvature and torsion. 
\begin{proof}
The difference between the Lie derivatives and the gauge transformations \eqref{gt2}, restricted to the $\mathbb{R}^3$ sector, \textit{i.e.}, when $\alpha^{ij}=0$, provides
\begin{eqnarray}\label{isom1}
(\mathcal{L}_X-\delta_{_{\mathbb{R}^3}})\omega^{ij}&=&\mathrm{i}_X\Omega^{ij}+\nabla(\mathrm{i}_X\omega^{ij})\;,\nonumber\\
(\mathcal{L}_X-\delta_{_{\mathbb{R}^3}})e^i &=& \mathrm{i}_XT^i+\nabla\left(\mathrm{i}_Xe^i-\vartheta^i\right)-(\mathrm{i}_X{\omega^i}_k)e^k\;,
\end{eqnarray}
then the $\mathbb{R}^3$ sector essentially contains diffeomorphisms iff $\alpha^{ij}=\mathrm{i}_X\omega^{ij}$, $\vartheta^i=\mathrm{i}_Xe^i$ and the constraints
\begin{eqnarray}\label{rel2}
\mathrm{i}_X\Omega^{ij}&=&0\;,\nonumber \\
\mathrm{i}_XT^i&=&0\;,
\end{eqnarray}
are imposed. See also \cite{Achucarro:1987vz,Witten:1988hc,Assimos:2019yln}.
\end{proof}
\end{lemma}

\begin{remark}\label{remark.foliations}
If the field equations in Proposition \ref{lemmafeq2} are assumed, the constraints \eqref{rel2} are automatically satisfied. Otherwise, the constraints \eqref{rel2} will impose restrictions on the diffeomorphism allowed in the resulting gravity theory. These constraints can formally be associated with spacetime foliations \cite{Dufour:2005th,Lavau:2018th}. Hence, a partial breaking of diffeomorphisms is induced.
\end{remark}

\begin{corollary}\label{diff.break1}
Lemma \ref{theorem.isometr} imply on the partial breaking of diffeomorphisms of the action \eqref{eh1},
\begin{eqnarray}\label{sub1}
\mathbb{R}^3\equiv\mathrm{Diff}(3)\longmapsto\mathrm{Diff}(D)_X=\left\{X\in\mathrm{Diff}(3)\;\big|\left.\mathrm{i}_X\Omega^{ij}=0\;\;;\;\;\mathrm{i}_XT^i=0\right\},\;\right.
\end{eqnarray}
with $D\le3$, depending on the intersection of foliations induced by each constraint.
\end{corollary}
\begin{proof}
The validity of \eqref{isom1} are subjected to the constraints \eqref{rel2}. Each of these constraints defines a certain foliation of $M'$. Both foliations will intersect and can be solved to find a subspace $\mathfrak{X}_{fol}(M')\subset\mathfrak{X}(M')$. This subspace will define a smaller set of allowed diffeomorphisms enjoyed by the action \eqref{eh1}.
\end{proof}

\begin{remark}\label{remark.solutions}

It is important to highlight that the Lemma \eqref{theorem.isometr} is a more general result for the three possibilities of obtaining local isometries in $M'$:

1) When $R^{ij} = 0$ and $T^i = 0$. In this case we have complete diffeomorphism symmetry, but classical vacuum solutions must be employed, as shown by E. Witten in \cite{Witten:1988hc}.

2) When $X = 0$. Here the diffeomorphisms are completely broken. This is not a situation of interest to us as it annihilates any possibility of better understanding these diffeomorphisms.

3) When $\mathrm{i}_XR^{ij} = 0$ and $\mathrm{i}_XT^i = 0$. This approach is the most general and is in between situations (1) and (2) aforementioned. In this case, the diffeomorphisms are reduced to a subgroup and these constraints will be our objects of study.

\end{remark}

Finally, we have all ingredients to construct the bundle map $f:P\longmapsto P^\prime_X$, with $P^\prime_X=(SO(3),M^\prime,\mathrm{Diff}(D)_X)$ being the coframe bundle with broken diffeomorphism symmetry for the gravity action \eqref{eh1}. For the general results about bundle maps, see for instance \cite{Kobayashi:1963fg,Nakahara:2003nw}.

\begin{theorem}\label{theo.cofr1}
The map $f:P\longmapsto P_X^\prime$ between the gauge bundle $P$ and the coframe bundle $P_X^\prime$ is realized by a series of maps:
\begin{eqnarray}
\Gamma_1&:&SU(2)\longmapsto SO(3)\;,\nonumber\\
\Gamma_2&:&\mathbb{R}^3\longmapsto\mathrm{Diff}(D)_X\;,\nonumber\\
\Gamma_3&:&M\longmapsto M^\prime\;,\nonumber\\
\Gamma_4&:&\{A,\theta\}\longmapsto\{\omega,e\}\;.\label{mapseries}
\end{eqnarray}
\end{theorem}

\begin{proof}
The schematics of the proof is pictorially represented in the Figure \ref{fig.diagram1}. The outline of the proof is provided as follows: The first step is to split $P$ in a sub-bundle $P_A=(SU(2)),M)_A$ and the annex $(\mathbb{R}^3)_\theta$. The sub-bundle $P_A$ is endowed with the connection $A$ because $SU(2)$ is a stability group of the Poincaré group $SU(2)\times\mathbb{R}^3$ \cite{Kobayashi:1963fg,McInnes:1984kz,Sobreiro:2010ji}. The annex $(\mathbb{R}^3)_\theta$ can be taken as a separated structure as a vector bundle obtained from $(SU(2)),M)_A$. From Propositions \ref{prop.mapgaugegrav} and \ref{prop.gaugemap}, the map $f_1:P_A\longmapsto(SO(3),M)_\omega$ is carried out. This map is a traditional bundle map between equivalent bundles \cite{Kobayashi:1963fg,Nakahara:2003nw} since $M\longmapsto M$. The annex $(\mathbb{R}^3)_\theta$ is then mapped, via Proposition \ref{prop.gaugemap} and Lemma \ref{theorem.isometr}, in the soldering form $e$ and the transition functions $t_{\alpha\beta}\in \mathrm{Diff}(D)_X$ characterizing the broken diffeomorphisms of $P^\prime_X$. Thus $f_2:(\mathbb{R}^3)_\theta\longmapsto(t_{\alpha\beta},e)$. The soldering form is then responsible for the map $\Gamma_3$ in \eqref{mapseries}. Finally, $P\longmapsto P^\prime_X$ is realized by "gluing" $(SO(3),M)_\omega$ with $(t_{\alpha\beta},e)$.
\end{proof}
\begin{figure}[H]
\centering
\begin{tikzcd}[column sep=normal]
& (SU(2),M)_{_{A}} \ar[r,mapsto,"f_1"] \ar[r] \ar[r] & (SO(3),M)_{_{\omega}}  \ar[dr]
&
& \\
P \ar[ur] \ar[dr]
&
&
& P'_{_{X}}\\
& (\mathbb{R}^3)_{_{\theta}} \ar[r,mapsto,"f_2"] & (\mathrm{t}_{\alpha\beta},e) \ar[ur]
\end{tikzcd}
\caption{Schematics of the map $P\longmapsto P^\prime_X$ in Theorem \ref{theo.cofr1} characterizing the map between a CS theory for the Poincaré group to an EH gravity in three dimensions.}\label{fig.diagram1}
\end{figure}
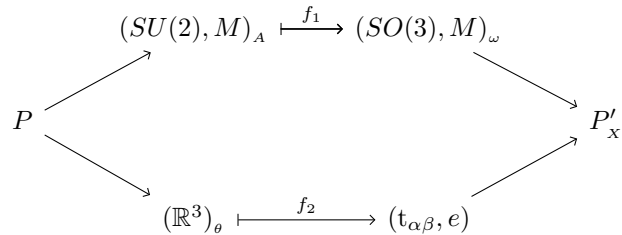

\subsection{Three-dimensional orthogonal Chern-Simons theory}\label{sub.cs3dortho}

As discussed in \cite{Achucarro:1987vz,Witten:1988hc,Assimos:2019yln}, the analysis can be extended to the principal bundle $P(SU(2)\times S(3),M)$ such that $SO(4)\equiv SU(2)\times S(3)$ with $S(3)$ being pseudo-translations. The algebra is modified only for the translational sector, namely $[Q_{i},Q_{j}]=\epsilon_{ijk} L^{k}$, while the Killing metric \eqref{iqf1} are exactly the same. Definitions \ref{def.connection} and \ref{def.action1} and Proposition \ref{prop.mapgaugegrav} can be used again to obtain
\begin{equation}\label{eh2}
S_{3grav}^\Lambda=-\frac{1}{16\pi G}\int\epsilon_{ijk}e^{i}\left(\Omega^{jk}-\frac{\Lambda^2}{3}e^{j} e^{k}\right)\;, 
\end{equation}
out from the CS action \eqref{cs1}. Clearly, action \eqref{eh2} is the EH action in the presence of cosmological constant $\Lambda^2=\mu^4$. The three-dimensional Newton's constant $G$ is the same as before. The action \eqref{eh2} is clearly invariant under $SO(3)$ gauge transformations. From Proposition \ref{lemma.gaugetransf2} and Lemma \ref{theorem.isometr}, applied for the $SO(4)$ gauge symmetry, the constraints to be demanded are now \cite{Assimos:2019yln}
\begin{eqnarray}\label{rel4}
\mathrm{i}_X\left(\Omega^{ij}-\Lambda^2e^ie^j\right)&=&0\;,\nonumber\\
\mathrm{i}_XT^i&=&0\;,
\end{eqnarray}
which is also satisfied by the field equations originated from the action \eqref{eh2}. In general, the diffeomorphisms will be reduced to

\begin{eqnarray}\label{sub1a}
S(3)\equiv\mathrm{Diff}(3)\longmapsto\mathrm{Diff}(D)_X^\Lambda=\left\{X\in \mathrm{Diff}(3)\;\big|\;\mathrm{i}_X\left(\Omega^{ij}-\Lambda^2e^ie^j\right)=0\;\;;\;\;\mathrm{i}_XT^i=0\right\},
\end{eqnarray}
with, again, $D\le3$.

Theorem \ref{theo.cofr1} can be adapted to the present case as follows:
\begin{theorem}\label{theo.cofr2}
The map $f:P\longmapsto P_X^\prime$ between the gauge bundle $P$ and the coframe bundle $P_X^\prime$ is realized by a series of maps:
\begin{eqnarray}
\Gamma_1&:&SU(2)\longmapsto SO(3)\;,\nonumber\\
\Gamma_2&:&S(3)\longmapsto\mathrm{Diff}(D)_X\;,\nonumber\\
\Gamma_3&:&M\longmapsto M^\prime\;,\nonumber\\
\Gamma_4&:&\{A,\theta\}\longmapsto\{\omega,e\}\;.\label{mapseries2}
\end{eqnarray}
\end{theorem}
\begin{proof}
The schematics of the proof is pictorially represented in the Figure \ref{fig.diagram2}. The outline of the proof follows the same prescription of Theorem \ref{theo.cofr1} and is omitted.
\end{proof}
\begin{figure}[H]
\centering
\begin{tikzcd}[column sep=normal]
& (SU(2),M)_{_{A}} \ar[r,mapsto,"f_1"] \ar[r] \ar[r] & (SO(3),M)_{_{\omega}}  \ar[dr]
&
& \\
P \ar[ur] \ar[dr]
&
&
& P'_{_{X}}\\
& S(3)_{_{\theta}} \ar[r,mapsto,"f_2"] & (\mathrm{t}_{\alpha\beta},e) \ar[ur]
\end{tikzcd}
\caption{Schematics of the map $P\longmapsto P^\prime_X$ characterizing the map between a CS theory for the orthogonal group $SO(4)$ group to an EH gravity with cosmological constant in three dimensions.}\label{fig.diagram2}
\end{figure}
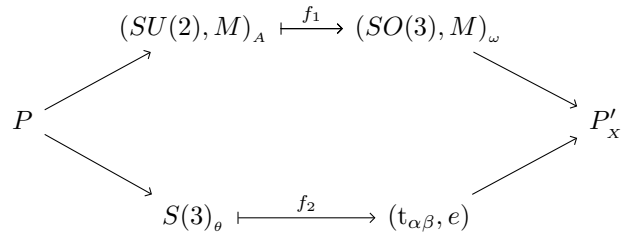

\subsection{Four-dimensional orthogonal Pontryagin theory}\label{sub.4d}

It is also possible to generalize the three-dimensional results to four dimensions for a Pontryagin topological gauge theory on the principal bundle $P=(SO(5),M)$ with the base space being a RC four-dimensional manifold $M$, as detailed in \cite{Assimos:2019yln}. The group decomposes as\footnote{One could equally start from the representation $SO(5)\equiv SU(2)\times SU(2)\times S(4)$.} $SO(5)\equiv SO(4)\times S(4)$ where $S(4)$ is the four-dimensional version of the pseudo-translations.

\begin{definition}\label{def.alg4d}
Let $\mathfrak{so}(5)$ be a semisimple Lie algebra of the orthogonal group in the representation $SO(5)\equiv SO(4)\times S(4)$, then the Killing form is nondegenerate and given by the metric tensor\footnote{Lower case Latin indices from $a$ to $h$ vary through $\{0,1,2,3\}$.} $\eta^{ab}=\mathrm{diag}(1,1,1,1)$. Thus, the algebra reads
\begin{eqnarray}\label{alg4d}
\left[L^{ab},L^{cd}\right]&=&-\frac{1}{2}\left[\left(\eta^{ac}L^{bd}+\eta^{bd}L^{ac}\right)-\left(\eta^{ad}L^{bc}+\eta^{bc}L^{ad}\right)\right]\;,\nonumber\\
\left[L^a,L^b\right]&=&-\frac{1}{2}L^{ab}\;,\nonumber\\
\left[L^{ab},L^c\right]&=&-\frac{1}{2}\left(\eta^{ac}L^b-\eta^{bc}L^a\right)\;,
\end{eqnarray}
where $L^{ab}$ are the $SO(4)$ antisymmetric generators and $L^a$ the generators of the $S(4)$ sector. 
\end{definition}

From Definitions \ref{def.connection} and \ref{def.fieldstrenght}, the gauge field and corresponding curvature 2-form are respectively given by $H=A^a_{\phantom{a}b}L_a^{\phantom{a}b}+\theta^aQ_a$ and $K=\left(F^a_{\phantom{a}b}-\frac{1}{4}\theta^{a}\theta_{b}\right)L_{a}^{\phantom{a}b}+\Pi^{a}Q_{a}$,
where $F^{a}_{\phantom{a}b}=dA^a_{\phantom{a}b}+A^a_{\phantom{a}c}A^c_{\phantom{c}b}$ and $\Pi^{a}=d\theta^{a}+ A^a_{\phantom{a}b}\theta^{b}$.  

\begin{definition}\label{def.gaugetransf3}
The infinitesimal gauge transformations $\delta_g$ are
\begin{eqnarray}\label{gt4}
\delta_g A^{ab}&=&D\alpha^{ab}+\frac{1}{4}\left(\theta^a\alpha^b-\theta^b\alpha^a\right)\;, \nonumber \\
\delta_g\theta^a&=&D\zeta^a-\frac{1}{2}\alpha^{ac}\theta_c\;, 
\end{eqnarray}
with $\varrho=\alpha^{ab}L_{ab}+\zeta^aL_a$ being the infinitesimal gauge parameter of the gauge group and $D$ is the covariant derivative with respect to the $SO(4)$ sector.
\end{definition}

\begin{definition}\label{def.pontryagin1}
The Pontryagin action providing dynamics to the gauge fields is
\begin{eqnarray}
S_{P}&=&\frac{1}{\kappa^2}\mathrm{Tr}\int KK\nonumber\\
&=&\frac{1}{2\kappa^2}\int\left[F^a_{\phantom{a}b}F_a^{\phantom{a}b}+\frac{1}{2}\left(\Pi^{a}\Pi_{a}-F^a_{\phantom{a}b}\theta_a\theta^b\right)\right]\;.
\end{eqnarray}
\end{definition}

The first term, $FF$ is the $SO(4)$ Pontryagin term, which can be written as a boundary term (the exterior derivative of the CS three-form). The rest can also be cast as a boundary term, $\displaystyle\int\left(\Pi^{a}\Pi_{a}-F^a_{\phantom{a}b}\theta_a\theta^b\right)=\displaystyle\int\;d\left(\Pi^a\theta_a\right)$. Thence, there are no field equations for the Pontryagin action. Clearly, the Pontryagin action is invariant under gauge transformations \eqref{gt4}.

\begin{proposition}\label{prop.4d1}
Let $SO(4)\longmapsto SO(4)$ be an identity map. Then the following changes in field variables are allowed
\begin{eqnarray}\label{id3}
A^{a}_{\phantom{a}b} &\longmapsto& \omega^{a}_{\phantom{a}b}\;,\nonumber\\
\theta^{a} &\longmapsto& \gamma e^{a}\;,
\end{eqnarray}
where $\gamma$ is an arbitrary mass parameter.
\end{proposition} 

\begin{proposition}
The four-dimensional gravity action obtained from applying the map \eqref{id3} into the action \eqref{def.pontryagin1} is
\begin{equation}
S_{Pgrav}=\frac{1}{16\pi G}\int\left[\frac{1}{2\Lambda^2}\Omega^a_{\phantom{a}b}\Omega_a^{\phantom{a}b}+T^{a}T_{a}-\Omega^a_{\phantom{a}b}e_ae^b\right]\;,\label{pontryagin2}
\end{equation}
with $G=\dfrac{\kappa^2}{4\pi\gamma^2}$ and $\Lambda^2=\dfrac{\gamma^2}{4}$. The curvature and torsion 2-forms are given by $\Omega^a_{\phantom{a}b}=d\omega^a_{\phantom{a}b}+\omega^a_{\phantom{a}c}\omega^c_{\phantom{c}b}$ and $T^a=De^a$ with $De^a=de^a+\omega^a_{\phantom{a}b}e^b$.
\end{proposition}

The first term in the action \eqref{pontryagin2} is the gravitational Pontryagin term. The rest compose the Nieh-Yan topological term \cite{Nieh:1981ww,Nieh:2007zz,Nieh:2018rlg}. The resulting theory is then a topological gravity in four-dimensional spacetime.

Proposition \ref{lemma.gaugetransf2} and Lemma \ref{theorem.isometr} can be generalized for the $SO(5)$ gauge group in a four-dimensional RC manifold. The corresponding constraints are now
\begin{eqnarray}\label{rel9}
\mathrm{i}_X\left(\Omega^{ab}-\Lambda^2e^ae^b\right)&=&0\;,\nonumber\\
\mathrm{i}_XT^a&=&0\;.
\end{eqnarray}
Thence, the resulting gravity theory enjoys broken diffeomorphism symmetry, namely
\begin{eqnarray}\label{sub1b}
\mathrm{Diff}(4)\longmapsto\mathrm{Diff}(D)^\Lambda_X=\left\{X\in \mathrm{Diff}(4)\;\big|\;\mathrm{i}_X\left(\Omega^{ab}-\Lambda^2e^ae^b\right)=0\;\;;\;\;\mathrm{i}_XT^a=0\right\}\;.
\end{eqnarray}
In the case of \eqref{sub1b}, $D\le4$. An interesting remark here is that, in contrast to the three-dimensional cases, constraints \eqref{rel9} are not on-shell satisfied.

Theorem \ref{theo.cofr1} describing the map from $P$ to $P^\prime_X=(SO(4),M^\prime,\mathrm{Diff}(D)_X^\Lambda)$ can be generalized for the present case as well:
\begin{theorem}\label{theo.cofr3}
The map $f:P\longmapsto P_X^\prime$ between the gauge bundle $P$ and the coframe bundle $P_X^\prime$ is realized by a series of maps:
\begin{eqnarray}
\Gamma_1&:&SO(4)\longmapsto SO(4)\;,\nonumber\\
\Gamma_2&:&S(4)\longmapsto\mathrm{Diff}(D)_X\;,\nonumber\\
\Gamma_3&:&M\longmapsto M^\prime\;,\nonumber\\
\Gamma_4&:&\{A,\theta\}\longmapsto\{\omega,e\}\;.\label{mapseries3}
\end{eqnarray}
\end{theorem}
\begin{proof}
The schematics of the proof is pictorially represented in the Figure \ref{fig.diagram3}. The outline of the proof follows the same prescription of Theorem \ref{theo.cofr1} and is omitted.
\end{proof}

\begin{figure}[H]
\centering
\begin{tikzcd}[column sep=normal]
& (SO(4),M)_{_{A}} \ar[r,mapsto,"f_1"] \ar[r] \ar[r] & (SO(4),M)_{_{\omega}}  \ar[dr]
&
& \\
P \ar[ur] \ar[dr]
&
&
& P'_{_{X}}\\
& S(4)_{_{\theta}} \ar[r,mapsto,"f_2"] & (\mathrm{t}_{\alpha\beta},e) \ar[ur]
\end{tikzcd}
\caption{Schematics of the map $P\longmapsto P^\prime_X$ characterizing the map between a Pontryagin theory for the orthogonal group $SO(5)$ group to an $SO(4)$ topological gravity in four dimensions.}\label{fig.diagram3}
\end{figure}
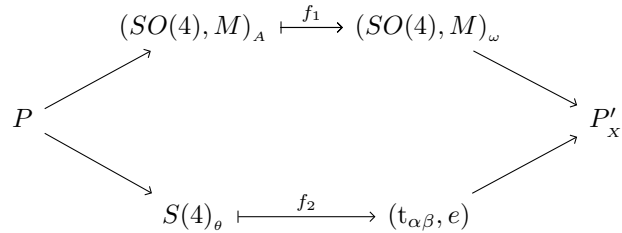

The summary of this entire Section is displayed in Table~\ref{table:1} below.

\begin{table}[H]
\scalefont{0.9}
\addtolength{\tabcolsep}{-2.0pt}
\begin{center}
\begin{tabular}{@{}c c c c c @{}}
\toprule
Dim. & Gauge symmetry & Gauge $\mapsto$ Gravity & Isometry  & Diff.\\
\midrule
3    & $SU(2)\times\mathbb{R}^3$   & CS $\mapsto$ EH  &  $SO(3)$ & $\mathrm{Diff}(D\le3)_X$  \\
3    & $SU(2)\times S(3)$   & CS $\mapsto$ EH+cc & $SO(3)$ & $\mathrm{Diff}(D\le3)^\Lambda_X$  \\
4   & $SO(5)$    & P $\mapsto$ P+NY  & $SO(4)$ & $\mathrm{Diff}(D\le4)^\Lambda_X$  \\
\bottomrule
\end{tabular}
\caption{Gauge-gravity equivalences. In the third column, we have the following keys: CS -- Chern-Simons; EH -- Einstein-Hilbert; cc -- Cosmological constant; P -- Pontryagin; NY -- Nieh-Yan.}
\label{table:1}
\end{center}
\end{table}

\section{Unconstrained interior homology}\label{sec.inthomol}

We now turn to the second part of the paper: the study the homology of the interior derivative operator in RC manifolds. We will first show that the interior homology for unconstrained RC manifolds is trivial. Then, in the next section, by considering constraints \eqref{rel2}, \eqref{rel4} and \eqref{rel9}, nontrivial interior homology groups will appear in RC manifolds. 

We first establish some further important concepts and definitions\footnote{In this section, all discussions concern RC manifolds of arbitrary dimension $n$. Moreover, with no confusion expected from the reader, we employ the language of simplicial/singular homology for interior homology.}:
\begin{definition}\label{def.chain}
The chain complex $(\mathcal{C}_q,\hat{\delta})$ is given by $\mathcal{C}_q = \mathcal{V}^q(M)$ and the differential $\hat{\delta}$ defines the homomorphism $\hat{\delta}:\mathcal{C}_q\to \mathcal{C}_{q-1}$. The differential $\hat{\delta}=\mathrm{i}_X$ is the interior derivative associated to $X\in \mathfrak{X}(M)$ such that $\mathrm{i}^2_X=0$. Then its homology groups are denoted by $H_q(M,X):=Z_q(M,X)/B_q(M,X)$, where $Z_q(M,X)$ is the set of closed $q$-forms called the $q$th cycle group and ${B_q(M,X)}$ is the set of exact $q$-forms called the $q$th boundary group.
\end{definition}

\begin{proposition}\label{prop.solu}
Let $\mathrm{i}_X$ act on a generic closed $q$-form\footnote{The index $A$ classify tensors in group space while $q$ classifies form ranks in the spacetime base manifold.} $\Delta^{A(q)}$,
\begin{equation}\label{homo1}
\mathrm{i}_X\Delta^{A(q)}=0\;.
\end{equation}
The solution of \eqref{homo1} for $\Delta^{A(q)}$ consists in two parts: a closed nontrivial part; and an exact trivial part:
\begin{equation}\label{homo2}
\Delta^{A(q)}=\Delta_o^{A(q)}+\mathrm{i}_X\Delta^{A(q+1)}\;,
\end{equation}
where $\Delta_o^{A(q)}$ is a closed non-exact $q$-form such that
\begin{equation}
\mathrm{i}_X\Delta_o^{A(q)}=0\;\Big|\;\Delta_o^{A(q)}\neq\mathrm{i}_X\widetilde{\Delta}^{A(q+1)}\;,
\end{equation}
with $\widetilde{\Delta}^{A(q+1)}$ and $\Delta^{A(q+1)}$ being $(q+1)$-forms.
\end{proposition}

\begin{proposition}\label{prop.solu2}
Interior homology as described by \eqref{homo1} and \eqref{homo2} takes place in the space of forms obeying covariance, polynomial locality, not depending on the Hodge dual, and that can be constructed from the equivalence classes of $q$-cycles,
\begin{equation}\label{cycle}
H^A_q(M,X)\equiv\{[\Delta^{A(q)}_o]|\Delta^{A(q)}_o\in Z^A_q(M,X)\}\;,     
\end{equation}
where each equivalence class $[\Delta^{A(q)}_0]$ is an homology class such that $\Delta^{A(q)}_0$ are obtained from the set of geometrical $q$-forms $\Phi=\{e,\omega,\Omega,T\}$ and covariant combinations between them. Moreover, the tangent space metric $\eta$ and the Levi-Civita tensor $\epsilon$ are at our disposal. 
\end{proposition}

An important comment can be made. All requirements we made on the construction of the $q$-cycle are of physical nature: The gravity actions we are considering, namely \eqref{eh1}, \eqref{eh2} and \eqref{pontryagin2}, are polynomial in the fields and do not depend on the Hodge dual. The spacetime dynamics depend only on the geometrical fields contained in $\Phi$ and no other $q$-forms. Thus, the geometrical properties of spacetime can only be determined due to these fields. Therefore, for physical reasons, the interior homology can only be affected by the fields in $\Phi$. We also point out that such prescription follows the usual BRST cohomology analysis in gauge theories \cite{Piguet:1995er}.

Solving \eqref{homo2} would give at least whether $H_q^A(M,X)$ is trivial or not. For RC manifolds, we should find $H_q^A(M,X)={0},\;\forall\;q$, where $0$ is the trivial group. In fact, to show that $H^A_q(M,X)=0$ one must prove that $\Delta_o^{A(q)}=0\;\forall\;A\;\mathrm{and}\;q$. This can be systematically checked for all possible cases. Nevertheless, the proof follows the same algorithm that we employ in the next section for the nontrivial constrained case. Hence, we omit the formal proof here, except for some explicit examples. The rest follows from the proofs of the next Section. But first, let us enunciate the result as a theorem:
\begin{theorem}\label{H01}
In Riemann-Cartan manifolds of any dimension, the interior homology groups are trivial, $H^A_q(M,X)=0$ \; $\forall$ $q$.
\end{theorem}

Let us proceed with some explicit examples to illustrate Theorem \ref{H01}. We start with, for instance, $H^A_0$ in three and four dimensions. In these cases, we must construct $\Delta^{A(0)}_o(M,X)$ out from the geometric space $\Phi$. Since in $\Phi$ there are no $0$-forms and we are not considering the Hodge operator, we can not construct any term for $\Delta^{A(0)}_o(M,X)$. Thus, $H^A_o(M,X)=0$, trivially, in three and four dimensions. A second example, which is not trivially immediate is $H^i_2(M,X)$ in three dimensions. In that case, one can infer that the most general form of $\Delta^{i(2)}_o(M,X)$ which is polynomially local in the geometric fields of $\Phi$ and covariant under local isometry transformations is given by
\begin{equation}
    \Delta^{ij(2)}_o(M,X)=a_1e^ie^j+a_2R^{ij}+a_3\epsilon^{ij}_{\phantom{ij}k}T^k\;,
\end{equation}
with $a_1$, $a_2$ and $a_3$ $\in\mathbb{R}$. Since any of these terms are invariant under the action of $\mathrm{i}_X$, we have that $H^{ij}_2(M,X)=0$ in three dimensions. To end the trivial examples, we consider $H^a_3(M,X)$ in four dimensions. In this case,
\begin{equation}
    \Delta^{a(3)}_o(M,X)=a_1\epsilon^a_{\phantom{a}bcd}e^be^ce^d+a_2\epsilon^a_{\phantom{a}bcd}R^{bc}e^d\;,
\end{equation}
with $a_1$ and $a_2$ $\in\mathbb{R}$. Again, since any term is invariant under the action of the interior derivative, we have $H^a_3(M,X)=0$ in four dimensions.

To end this section we must remark that $H_o^A(M,X)=0$ only because we have no $0$-forms in $\Phi$. If $0$-forms were considered and since $\mathrm{i}_X$ annihilates $0$-forms, $H_0^A(M,X)$ would be nontrivial.

\section{Constrained interior homology}\label{sec.const.inthomol}

In this section, we show that the effect of the constraints on the spacetime diffeomorphisms (see Section \ref{sec.gge})) is to make the interior derivative homology nontrivial. We split the analysis in the three and four dimensional cases.

\subsection{Three-dimensional Poincaré theory}\label{cs3d}

The general classes of interior homologies to be computed in this subsection are 
\begin{equation}
H^{A}_{q}(M,X)=Z_{q}^{A}(M,X)/B_{q}^{A}(M,X)\label{def.homo},
\end{equation}
with $A\equiv ijk\cdots\;\big|\;i,j,k,\ldots\in\{0,1,2\}$ and $q\in\{0,1,2,3\}$.  The homology problem is defined by \eqref{homo1} and the solution is formally given by \eqref{homo2}. From explicitly computations, we will compute the most relevant $q$-cycles $\Delta_o^{A(q)}$ and determine the associated homology groups $H^A_q(M,X)$. Let us start by defining the space of geometric forms we are interested in the three dimensional case, namely $\Phi=\{e^i,\omega^{ij},\Omega^{ij},T^i\}\subset\{\mathcal{V}^q(M)\}$.

The gravitational theory considered here is the EH action \eqref{eh1} and the constraints  \eqref{rel2} must be invoked. Thence, curvature and torsion are restricted to a subset of closed 2-forms respecting \eqref{rel2}. 

\begin{theorem}\label{theo1.triv}
The homology group $H_0^A(M,X)={0}$. 
\end{theorem}

\begin{proof}\label{proof.theo1}
Clearly, since there are no 0-forms in $\Phi$ (and no Hodge dual is allowed) we actually have $\Delta_o^{A(0)}(M,X)=0$ and thus 
\begin{equation}
H_0^A(M,X)={0}\;.\label{homo0A1}
\end{equation}
\end{proof}

\begin{theorem}\label{theo2.triv}
The homology group $H_1^A(M,X)={0}$.
\end{theorem}
\begin{proof}\label{proof.theo2}
The only possibilities of 1-forms in $\Phi$ are $e^i$, $\omega^{ij}$, and their combinations with $\eta_{ij}$ and $\epsilon_{ijk}$. Thus, $\Delta_o^{A(1)}$ will inevitably be linear combinations of the 1-form fields in $\Phi$. For instance, at $\Delta_o^{(1)}$, there are no possibilities since there are no combinations that can be constructed with no group indices. At $\Delta_o^{i(1)}$ we have only one term
\begin{equation}
\Delta_o^{i(1)}=a_1e^i\;,
\end{equation}
where $a_1\in\mathbb{R}$. From covariance requirement, the term $\epsilon^i_{\phantom{i}jk}\omega^{jk}$ is automatically ruled out. From the fact that $\mathrm{i}_Xe^i\neq0$, we get $a_1=0$. Thus, $\Delta_o^{i(1)}=0$. The same analysis hold for all possible $A$. Therefore,
\begin{equation}
H_1^A(M,X)={0}\;.\label{homo1A1}
\end{equation}
\end{proof}

\begin{theorem}\label{theo2.nontriv}
The homology group $H_2^A(M,X)\ne{0}$ if $\mathrm{dim}\,A>0$.
\end{theorem}
\begin{proof}\label{proof.theo2nt}
The covariant 2-forms we can construct out from $\Phi$ are $e^ie^j$, $\Omega^{ij}$, $T^i$, and their combinations with $\eta_{ij}$ and $\epsilon_{ijk}$. But only $\Omega^{ij}$ and $T^i$ are invariant under the action of the interior derivative. First, we observe that there is no way to construct 2-forms out of these possibilities with no group index. Hence, $H_2(M,X)={0}$. Then, it is easy to check that the first four nontrivial 2-cycles are given by
\begin{eqnarray}
\Delta^{i(2)}_o&=&a_1T^i+a_2\epsilon^i_{\phantom{1}jk}\Omega^{jk}\;,\nonumber\\
\Delta^{ij(2)}_o&=&\epsilon^{ij}_{\phantom{ij}k}\Delta_o^{k(2)}\nonumber\\
&=&b_1\Omega^{ij}+b_2\epsilon^{ij}_{\phantom{ij}k}T^k\;,\nonumber\\
\Delta^{ijk(2)}_o&=&\eta^{ij}\Delta_o^{k(2)}+\epsilon^{ij}_{\phantom{ij}l}\Delta_o^{lk(2)}+\mathrm{prm.}\nonumber\\
&=&c_1\epsilon^{ij}_{\phantom{ij}l}\Omega^{lk}+c_2\eta^{ij}T^k+c_3\eta^{ij}\epsilon^k_{\phantom{k}mn}\Omega^{mn}+\mathrm{prm.}\;,\nonumber\\
\Delta^{ijkl(2)}_o&=&\epsilon^{ijk}\Delta_o^{l(2)}+\eta^{ij}\epsilon^{kl}_{\phantom{kl}m}\Delta_o^{m(2)}+\eta^{ij}\Delta_o^{kl(2)}+\mathrm{prm.}\nonumber\\
&=&d_1\epsilon^{ijk}T^l+d_2\eta^{ij}\epsilon^{kl}_{\phantom{kl}m}T^m+d_3\eta^{ij}\Omega^{kl}+d_4\epsilon^{ijk}\epsilon^l_{\phantom{l}mn}\Omega^{mn}+\mathrm{prm.}\;,\label{homo2A}
\end{eqnarray}
with "$\mathrm{prm}.$" defining all possible independent index permutations and the arbitrary coefficients $a_I,b_I,c_I,d_I \in\mathbb{R}$. Essentially, all 2-cycles will be determined from the fundamental nontrivial 2-cycle $\Delta_o^{i(2)}$. Hence, from \eqref{homo2A} one can infer the corresponding homology groups by counting the number of independent coefficients $\in\mathbb{R}$,
\begin{eqnarray}
    H_2^i(M,X)&=&\mathbb{R}\oplus\mathbb{R}\;,\nonumber\\
    H_2^{ij}(M,X)&=&\mathbb{R}\oplus\mathbb{R}\;,\nonumber\\
    H_2^{ijk}(M,X)&=&\bigoplus_1^9\mathbb{R}\;,\nonumber\\
    H_2^{ijkl}(M,X)&=&\bigoplus_1^{22}\mathbb{R}\;.
    \label{homo2Aa}
\end{eqnarray}
\end{proof}

\begin{theorem}\label{theo4.triv}
The homology group $H_3^A(M,X)={0}$.
\end{theorem}

\begin{proof}\label{proof.theo4}
Let us first notice that, as in any homology problem, the case $q=n$ ($n=3$ in the specific case) is of particular interest, because $B_3^A(M,X)={0}$ (the trivial part of $\Delta^{A(4)}$ vanishes since there are no 4-forms in three-dimensional manifolds). Thus $H_3^A(X,M)=Z^A_3(X,M)$. The possibilities at the nontrivial sector are now composed by the covariant 3-forms in $\Phi$: $e^ie^je^k$, $\Omega^{ij}e^k$, $T^ie^j$, the Chern-Simons 3-form $\omega^i_{\phantom{i}j}d\omega^j_{\phantom{j}i}+\frac{2}{3}\omega^i_{\phantom{i}j}\omega^j_{\phantom{j}k}\omega^k_{\phantom{k}i}$, and their combinations with $\eta_{ij}$ and $\epsilon_{ijk}$. Non of them are invariant under the action of the interior derivative. Thus, $H^A_3(M,X)={0}$.
\end{proof}

\subsection{Three-dimensional orthogonal theory}\label{sub.sub.extend}

Now we consider the action \eqref{eh2} which is invariant under $SO(4)$ symmetry. This action is subjected to the constraints \eqref{rel4}. From now on we omit the trivial cases because their proofs follow the same recipe of the previous cases. In fact, we may collect these results in one single theorem:
\begin{theorem}\label{theo5.triv}
The homology groups $H_0^A(M,X)=H_1^A(M,X)=H_2(M,X)=H_3^A(M,X)={0}$.
\end{theorem}

The nontrivial interior homology groups are collected in the next theorem:
\begin{theorem}\label{theo5.nontriv}
The homology groups $H_2^A(M,X)\ne{0}$ if $\mathrm{dim}\,A>0$.
\end{theorem}
\begin{proof}\label{proof.theo5nt}
Any 2-cycle $\Delta_o^{A(2)}$ will be a linear combination of the covariant 2-forms in $\Phi$, namely $e^ie^j$, $\Omega^{ij}$, $T^i$, and their combinations with $\eta_{ij}$ and $\epsilon_{ijk}$. The constraints \eqref{rel4} establishes that $T^i$ is invariant \emph{wrt} $\mathrm{i}_X$ and that $\Omega^{ij}$ and $e^ie^j$ are related under the action of the interior derivative. It can be systematically checked that the first four covariant nontrivial 2-cycles which are invariant under the action of the interior derivative are given by
\begin{eqnarray}
\Delta^{i(2)}_o&=&a_1T^i+a_2\epsilon^i_{\phantom{1}jk}\left(\Omega^{jk}-\Lambda^2e^je^k\right)\;,\nonumber\\
\Delta^{ij(2)}_o&=&\epsilon^{ij}_{\phantom{ij}k}\Delta_o^{k(2)}\nonumber\\
&=&b_1\left(\Omega^{ij}-\Lambda^2e^ie^j\right)+b_2\epsilon^{ij}_{\phantom{ij}k}T^k\;,\nonumber\\
\Delta^{ijk(2)}_o&=&\eta^{ij}\Delta_o^{k(2)}+\epsilon^{ij}_{\phantom{ij}l}\Delta_o^{lk(2)}+\mathrm{prm.}\nonumber\\
&=&c_1\epsilon^{ij}_{\phantom{ij}l}\left(\Omega^{lk}-\Lambda^2e^le^k\right)+c_2\eta^{ij}T^k+c_3\eta^{ij}\epsilon^k_{\phantom{k}mn}\left(\Omega^{mn}-\Lambda^2e^me^n\right)+\mathrm{prm.}\;,\nonumber\\
\Delta^{ijkl(2)}_o&=&\epsilon^{ijk}\Delta_o^{l(2)}+\eta^{ij}\epsilon^{kl}_{\phantom{kl}m}\Delta_o^{m(2)}+\eta^{ij}\Delta_o^{kl(2)}+\mathrm{prm.}\nonumber\\
&=&d_1\epsilon^{ijk}T^l+d_2\eta^{ij}\epsilon^{kl}_{\phantom{kl}m}T^m+d_3\eta^{ij}\left(\Omega^{kl}-\Lambda^2e^ke^l\right)+d_4\epsilon^{ijk}\epsilon^l_{\phantom{l}mn}\left(\Omega^{mn}-\Lambda^2e^me^n\right)+\mathrm{prm.}\;,\label{deltai2so4b}
\end{eqnarray}
with $a_I,b_I,c_I,d_I\;\in\mathbb{R}$. In practice, all extra terms depending on $e^ie^j$ will join to $\Omega^{ij}$ with a single independent parameter. Eventually, one concludes that relations \eqref{homo2Aa} remain valid.
\end{proof}

\subsection{Four-dimensional orthogonal theory}\label{ymfol4d}

A similar analysis can be performed for the four-dimensional case focusing on the constraints \eqref{rel9} acting on $q$-cycles defined from the subspace $\Phi=\{\Omega^{ab},T^a,\omega^{ab},e^a\}$, and the invariant tensors $\eta^{ab}$ and $\epsilon_{abcd}$ for the $SO(5)$ invariant topological action \eqref{pontryagin2}.  From the same reasons we discussed in the three dimensional cases, one can easily show the following collective theorem:
\begin{theorem}\label{theo5.trivb}
The homology groups $H_0^A(M,X)=H_1^A(M,X)=H_2(M,X)=H_3^A(M,X)={0}$.
\end{theorem}

For the nontrivial interior homology groups we have the following theorems:

\begin{theorem}\label{theo6.nontriv}
The homology groups $H_2^A(M,X)\ne{0}$ if $\mathrm{dim}\,A>0$.
\end{theorem}
\begin{proof}\label{proof.theo6nt}
Any 2-cycle $\Delta_o^{A(2)}$ will be a linear combination of the covariant 2-forms in $\Phi$, namely $e^ae^b$, $\Omega^{ab}$, $T^a$, and their combinations with $\eta_{ab}$ and $\epsilon_{abcd}$. The constraints \eqref{rel4} establishes that $T^a$ is invariant \emph{wrt} $\mathrm{i}_X$ and that $\Omega^{ab}$ and $e^ae^a$ are related under the action of the interior derivative. The first four covariant nontrivial 2-cycles which are invariant under the action of $\mathrm{i}_X$ are (See also Theorem \ref{theo5.nontriv})
\begin{eqnarray}
\Delta^{a(2)}_o&=&a_1T^a\;,
\end{eqnarray}
\begin{eqnarray}
\Delta^{ab(2)}_o&=&b_1\left(\Omega^{ab}-\Lambda^2e^ae^b\right)+b_2\epsilon^{ab}_{\phantom{ab}cd}\left(\Omega^{cd}-\Lambda^2e^ce^d\right)\;,
\end{eqnarray}
\begin{eqnarray}
\Delta^{abc(2)}_o&=&\eta^{ab}\Delta_o^{c(2)}+\epsilon^{abc}_{\phantom{abc}d}\Delta_o^{d(2)}+\mathrm{prm.}\nonumber\\
&=&c_1\eta^{ab}T^c+c_2\epsilon^{abc}_{\phantom{abc}d}T^d+\mathrm{prm.}\;,
\end{eqnarray}
and
\begin{eqnarray}
\Delta^{abcd(2)}_o&=&\eta^{ab}\Delta_o^{cd(2)}+\epsilon^{abc}_{\phantom{abc}e}\Delta_o^{ed(2)}+\mathrm{prm.}\nonumber\\
&=&d_1\eta^{ab}\left(\Omega^{cd}-\Lambda^2e^ce^d\right)+d_2\eta^{ab}\epsilon^{cd}_{\phantom{cd}ef}\left(\Omega^{ef}-\Lambda^2e^ee^f\right)+d_3\epsilon^{abc}_{\phantom{abc}e}\left(\Omega^{ed}-\Lambda^2e^ee^d\right)+\nonumber\\
&+&d_4\epsilon^{abc}_{\phantom{abc}e}\epsilon^{ed}_{\phantom{ed}fg}\left(\Omega^{fg}-\Lambda^2e^fe^g\right)+\mathrm{prm.}\;,\label{deltai2so4c}
\end{eqnarray}
with $a_I,b_I,c_I,d_I\;\in\mathbb{R}$. Clearly, all extra terms depending on $e^ae^b$ will join to $\Omega^{ab}$ with a single independent parameter. Eventually, one concludes that
\begin{eqnarray}
H_2^a(M,X)&=&\mathbb{R}\;,\nonumber\\
H_2^{ab}(M,X)&=&\mathbb{R}\oplus\mathbb{R}\;,\nonumber\\
H_2^{abc}(M,X)&=&\bigoplus_1^4\mathbb{R}\;,\nonumber\\
H_2^{abcd}(M,X)&=&\bigoplus_1^{18}\mathbb{R}\;.
\label{homo2Ab}
\end{eqnarray}
\end{proof}

\begin{theorem}\label{theo7.nontriv}
The homology groups $H_4^A(M,X)\ne{0}$.
\end{theorem}

\begin{proof}\label{proof.theo7nt}
Any 4-cycle $\Delta_o^{A(2)}$ will be a linear combination of the covariant 4-forms in $\Phi$, namely $e^ae^be^ce^d$, $\Omega^{ab}\Omega^{cd}$, $T^aT^b$, $\Omega^{ab}e^ce^d$, $\Omega^{ab}T^c$, $T^ae^be^c$, and their combinations with $\eta_{ab}$ and $\epsilon_{abcd}$. For instance, it can be checked that the first three\footnote{The fourth 4-cycle is omitted due to its excessive length.} covariant nontrivial 4-cycles, which are invariant under the action of $\mathrm{i}_X$ can be constructed out from $\Phi$, are
\begin{eqnarray}
\Delta_o^{(4)}&=&\eta_{ac}\eta_{bd}\Delta_o^{ab(2)}\Delta_o^{ cd(2)}+\eta_{ab}\Delta_o^{a(2)}\Delta_o^{b(2)}\nonumber\\
&=&a_1\left(\Omega^{ab}-\Lambda^2e^ae^b\right)\left(\Omega_{ab}-\Lambda^2e_ae_b\right)+a_2\epsilon_{abcd}\left(\Omega^{ab}-\Lambda^2e^ae^b\right)\left(\Omega^{cd}-\Lambda^2e^ce^d\right)+a_3T^aT_a\;,
\end{eqnarray}
\begin{eqnarray}
\Delta^{a(4)}_o&=&\eta_{bc}\Delta_o^{ab(2)}\Delta_o^{c(2)}\nonumber\\
&=&b_1\left(\Omega^{ab}-\Lambda^2e^ae^b\right)T_b+b_2\epsilon^a_{\phantom{a}bcd}\left(\Omega^{bc}-\Lambda^2e^be^c\right)T^d\;,
\end{eqnarray}
and
\begin{eqnarray}
\Delta^{ab(4)}_o&=&\eta^{ab}\Delta_o^{(4)}+\Delta_o^{a(2)}\Delta_o^{b(2)}+\epsilon^{ab}_{\phantom{ab}cd}\Delta_o^{c(2)}\Delta_o^{d(2)}+\eta_{cd}\Delta_o^{ac(2)}\Delta_o^{bd(2)}+\mathrm{prm.}\nonumber\\
&=&c_1\eta^{ab}\left(\Omega^{cd}-\Lambda^2e^ce^d\right)\left(\Omega_{cd}-\Lambda^2e_ce_d\right)+c_2\eta^{ab}\epsilon_{cdef}\left(\Omega^{cd}-\Lambda^2e^ce^d\right)\left(\Omega^{ef}-\Lambda^2e^ee^f\right)+\nonumber\\
&+&c_3\eta^{ab}T^cT_c+c_4\left(\Omega^{ac}-\Lambda^2e^ae^c\right)\left(\Omega_c^{\phantom{c}b}-\Lambda^2e_ce^b\right)+c_5\epsilon^{bc}_{\phantom{bc}de}\left(\Omega^{a}_{\phantom{a}c}-\Lambda^2e^ae_c\right)\left(\Omega^{de}-\Lambda^2e^de^e\right)+\nonumber\\
&+&\mathrm{prm.}\;,
\label{deltai2so4d}
\end{eqnarray}
with $a_I,b_I,c_I\;\in\mathbb{R}$. Clearly, all extra terms depending on $e^ae^b$ will join to $\Omega^{ab}$ with a single independent parameter. Eventually, one concludes that
\begin{eqnarray}
H_4(M,X)&=&\bigoplus_1^3\mathbb{R}\;,\nonumber\\
H_4^a(M,X)&=&\mathbb{R}\oplus\mathbb{R}\;,\nonumber\\
H_4^{ab}(M,X)&=&\bigoplus_1^6\mathbb{R}\;.
\label{homo2Ac}
\end{eqnarray}
\end{proof}

\section{Conclusions}\label{sec.concl}

In this work we mathematically formalized the results obtained in \cite{Assimos:2019yln} which, in turn, generalized the results described in \cite{Achucarro:1987vz,Witten:1988hc}. The essence of these results are contained in Theorems \ref{theo.cofr1}, \ref{theo.cofr2} and \ref{theo.cofr3}. Physically, these theorems establish that is possible to a diffeomorphic constrained gravity theory to be induced from topological gauge theories when the coframe bundle $P'_X$ for a gravity theory with reduced diffeomorphism invariance is construct from a principal bundle $P$ for a gauge theory. The reduced diffeomorphism invariance is a consequence of the constraints required for consistency of the map $P\longmapsto P'_X$.

In the second part of the paper we explored the consequences of the constraints in the homology of the interior derivative over RC manifolds. The explicit computation of nontrivial $q$-cycles for each class of constraints (three-dimensional gravity with and without cosmological constant and four-dimensional topological gravity) were performed in order to infer the respectively interior homology groups. We have found that the constraints bring several nontrivialities for the interior homology in RC manifolds for the three-dimensional gravity with no cosmological constant  \eqref{eh1}, the results are collected in Theorems \ref{theo1.triv} to \ref{theo4.triv}. For the three-dimensional gravity with nonvanishing cosmological constant \eqref{eh2}, the results are collected in Theorems \ref{theo5.triv} and \ref{theo5.nontriv}. For the four-dimensional topological gravity \eqref{pontryagin2}, the results are collected in Theorems \ref{theo5.trivb} to \ref{theo7.nontriv}.

Finally, we can point out some interesting studies for future investigation. The immediate question is how matter fields influence all results here obtained and how they affect the on-shell cases. The answer is not obvious because the introduction of matter fields may affect the gauge transformations and certainly affect the field equations. Another question is whether our results can be generalized to other spacetime dimensions. For instance, one different example that could be investigated is a $U(1)\times\mathbb{R}^2$ gauge theory in two dimensions. Such theory could be mapped in an Abelian two-dimensional gravity theory with $SO(2)$ local isometries \cite{Sobreiro:2021leg}. Moreover, CS theories can be defined in any odd spacetime dimension while Pontryagin theories can be defined in any even spacetime dimension.

\section*{Acknowledgments}

The authors are grateful to Thomas Endler and Renan Assimos, from Max Planck Institute for Mathematics in the Sciences, for the digital elaboration of Figures \ref{fig1} and \ref{fig2}. This study was financed in part by The Coordenação de Aperfeiçoamento de Pessoal de Nível Superior -- Brasil (CAPES) -- Finance Code 001.

\bibliography{BIB}
\bibliographystyle{utphys2}

\end{document}